\documentclass{amsart}
\usepackage{euscript}
\usepackage{pstricks}

\usepackage{graphics}
\usepackage{epsfig}
\usepackage{boxedminipage}
\usepackage{url}

\newcommand{\F}{\mathbb F}
\newcommand{\Z}{\mathbb Z}
\newcommand{\Q}{\mathbb Q}
\newcommand{\N}{\mathbb N}

\newcommand{\Cal}{\EuScript}

\newcommand{\md}{\mathbin{\mathsf{mod}}}

\renewcommand{\:}{\colon}
\renewcommand{\>}{\rightarrow}

\renewcommand{\labelenumi}{\rm(\theenumi)}
\renewcommand{\theenumi}{\roman{enumi}}

\newtheorem{thm}{Theorem}[section] 
\newtheorem{prop}[thm]{Proposition}

\theoremstyle{definition}

\theoremstyle{remark}

\theoremstyle{plain}

\newtheorem*{OpQu*}{Open question}
\theoremstyle{definition}

\newtheorem*{defn*}{Definition}
\theoremstyle{remark}
\newtheorem{note}[thm]{Note}
\newtheorem*{note*}{Note}

\newtheorem*{exmp*}{Example}
\newtheorem*{exmps*}{Examples}

\begin{document}

\title[Automata finiteness criterion]{Automata finiteness criterion in terms of van der Put series of  automata
functions}
\author{Vladimir Anashin}
\thanks{Supported in
parts by Russian Foundation for Basic Research grant No 09-01-00653-a and by Chinese Academy of Sciences visiting professorship for senior international
scientists grant No  2009G2-11}
\subjclass[2000]{Primary 11E95; Secondary 11B85, 68Q70}

\keywords{Automaton, finiteness conditions, $p$-adic numbers, van der Put
series,  automata sequences}
\address{
Institute for Information Security\\
Moscow State University\\
Leninskie Gory, 1\\
119991 Moscow\\
Russia}
\email{anashin@iisi.msu.ru; vladimir.anashin@u-picardie.fr}

\dedicatory{To Professor Igor Volovich on the occasion of his 65-th birthday}
\maketitle

\begin{abstract}
In the paper we develop the $p$-adic theory of discrete automata. Every automaton
$\mathfrak A$
(transducer) whose input/output alphabets consist of $p$ symbols can be associated
to a continuous (in fact, 1-Lipschitz) map from $p$-adic integers to $p$ integers, the automaton function $f_\mathfrak A$. The $p$-adic theory (in
particular, the $p$-adic ergodic theory) turned out to be very efficient
in a study of properties of automata expressed via properties of automata
functions. In the paper we prove a criterion for finiteness of the number
of states of automaton in terms of van der Put series of the automaton function.
The criterion displays  connections between $p$-adic analysis
and the theory of automata sequences.
\end{abstract}

\section{Introduction}
\label{sec:Intro}
The $p$-adic numbers, which appeared more than a century ago in Kurt Hensel's works as a pure mathematical construction, see e.g. \cite{Hensel}, at the
end of XX century were recognized as a base for adequate descriptions of physical, biological, cognitive and information processing phenomena. The pioneer papers in these studies were works of Vladimirov and Volovich \cite{VL}, \cite{VL1}, \cite{Volovich:1987} followed by 
monograph \cite{Vladimirov/Volovich/Zelenov:1994}. Although the papers (and
the monograph) are focused on application of the $p$-adic theory to mathematical
physics, the impact of these works was much wider than physical models only:
Inspired by these works, many scientists started applying $p$-adic methods to
their own areas of research. Now the $p$-adic theory, and wider, ultrametric
analysis and ultrametric dynamics, is a rapidly developing area that finds
applications to various sciences (physics, biology, genetics, cognitive sciences,
information sciences, computer science, cryptology, numerical methods, etc.). On the contemporary
state-of-the art, the interested
reader is referred to the monograph \cite{AnKhr} and references therein.
The current paper concerns application of $p$-adic methods to automata theory,
the both mathematical and applied science which has numerous applications in
engineering sciences, linguistics, computer science, etc. The paper displays
tight connections between the $p$-adic theory and the theory of automata sequences. Note also that automata may be regarded as a mathematical model
of the `causality law', \cite{Vuillem_DigNum}; so the present paper may have some relations to physics as well.

By the definition, the (non-initial) automaton is 5-tuple $\mathfrak A=\langle\Cal I,\Cal S,\Cal O,S,O\rangle$
where $\mathcal I$ is a finite set, the \emph{input alphabet}; $\Cal O$  is a finite set, the \emph{output alphabet}; $\Cal S$ is a non-empty (possibly, infinite) set of \emph{states}; 
$S\colon\Cal I\times\Cal S\to \Cal S$ is  a \emph{state transition function}; $O\colon\Cal I\times\Cal S\to \Cal O$ is an \emph{output function}.  The automaton
where both input alphabet $\Cal I$ and output alphabet $\Cal O$ are non-empty
is called the \emph{transducer}, see e.g.
\cite{Allouche-Shall}; the automaton where the input alphabet is empty whereas
the output alphabet is not empty is
called the \emph{generator}.
The \emph{initial automaton} $\mathfrak A(s_0)=\langle\Cal I,\Cal S,\Cal O,S,O, s_0\rangle$ is an automaton
$\mathfrak A$ where one state $s_0\in\Cal S$ is fixed; it is called the \emph{initial
state}.  We
stress that the definition of  the initial automaton $\mathfrak A(s_0)$ is nearly the same as the one of
\emph{Mealy automaton} (see e.g. \cite{Bra})
with the only important difference: 
the set of states
$\Cal S$ of $\mathfrak A(s_0)$ is \emph{not necessarily finite}.

Given a non-empty alphabet $\Cal A$, its elements are called \emph{symbols},
or \emph{letters}. By the definition, a \emph{word of length $n$ over alphabet $\Cal A$} is a finite sequence (stretching from
right to left)
$\alpha_{n-1}\cdots\alpha_1\alpha_0$, where $\alpha_{n-1},\ldots,\alpha_1,\alpha_0\in\Cal
A$. The \emph{empty word} is a sequence
of length 0, that is, the one that contains no symbols. Hereinafter the length
of the word $w$ is denoted via $|w|$. Given a word $w=\alpha_{n-1}\cdots\alpha_1\alpha_0$,
any word $v=\alpha_{k-1}\cdots\alpha_1\alpha_0$, $k\le n$, is called a \emph{prefix}
of the word $w$; whereas any word $u=\alpha_{n-1}\cdots\alpha_{i+1}\alpha_i$,
$0\le i\le n-1$ is called a \emph{suffix} of the word $w$. Given words $a=\alpha_{n-1}\cdots\alpha_1\alpha_0$
and $b=\beta_{k-1}\cdots\beta_1\beta_0$, the
\emph{concatenation} $a\circ b$ is the following word (of length $n+k$):
\[
a\circ b=\alpha_{n-1}\cdots\alpha_1\alpha_0\beta_{k-1}\cdots\beta_1\beta_0.
\] 

Given an input word $w=\chi_{n-1}\cdots\chi_1\chi_0$ over the alphabet $\Cal
I$, an initial transducer $\mathfrak A(s_0)=\langle\Cal I,\Cal S,\Cal O,S,O, s_0\rangle$
transforms $w$ to output word $w^\prime=\xi_{n-1}\cdots\xi_1\xi_0$ over the
output alphabet $\Cal O$ as follows (cf. Figure \ref{fig:Transd-sc}): 
Initially the transducer  $\mathfrak A(s_0)$ is at the state
$s_0$; accepting the input symbol $\chi_0\in\Cal I$, the transducer outputs the
symbol $\xi_0=O(\chi_0,s_o)\in\Cal O$ and reaches the state
$s_1=S(\chi_0,s_0)\in\Cal S$; then the transducer accepts the next input symbol
$\chi_1\in\Cal I$, reaches the state
$s_2=S(\chi_1,s_1)\in\Cal S$, outputs $\xi_1=O(\chi_1,s_1)\in\Cal O$, and the routine repeats.
\begin{figure}[h]
\begin{quote}\psset{unit=0.5cm}
 \begin{pspicture}(1,2.5)(24,13)
\pscircle[fillstyle=solid,fillcolor=yellow,linewidth=2pt](12,10){1}
\pscircle[fillstyle=solid,fillcolor=yellow,linewidth=2pt](12,4){1}
  \psline(12,11)(12,12)
  \psline{->}(18,7)(13,7)
  \psline{->}(6,10)(11,10)
  \psline{->}(6,4)(11,4)
  \psline{->}(4,7)(6,7)
  \psline(6,4)(6,10)
  \psline(18,12)(18,7)
  \psline(12,12)(18,12)
  \psline{->}(12,8)(12,9)
  \psline(12,3)(12,2)
  \psline{->}(12,2)(18,2)
  \psline{<-}(12,5)(12,6)
  \psframe[fillstyle=solid,fillcolor=yellow,linewidth=2pt](11,6)(13,8)
  \uput{0}[180](12.4,7){$s_i$}
  \uput{0}[180](3.5,7){$\cdots\chi_{i+1}\chi_i$}
  \uput{0}[90](12,9.7){$S$}
  \uput{1}[90](12,2.7){$O$}
  \uput{1}[0](11.7,11.3){$s_{i+1}=S(\chi_i,s_i)$}
 \uput{1}[0](11.7,12.7){{\sf state transition}}
 \uput{1}[0](0,8){{\sf input}}
  \uput{1}[0](11.5,2.5){$\xi_i=O(\chi_i,s_i)$}
  \uput{1}[0](17.5,2){$\xi_i\xi_{i-1}\cdots\xi_0$}
  \uput{1}[0](18,3){{\sf output}}
 \end{pspicture}
\end{quote}
\caption{Initial transducer, schematically}
\label{fig:Transd-sc}
\end{figure}
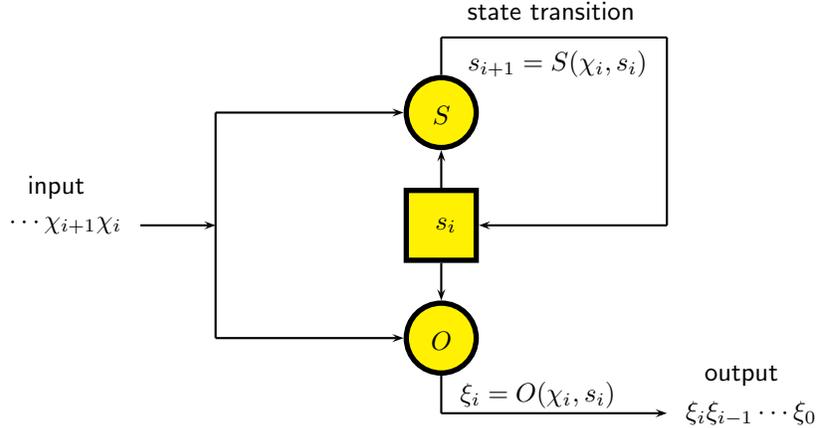
Throughout the paper, `automaton' mostly stands for `initial automaton';
 we make  corresponding remarks if not. Further in the paper we mostly
 consider transducers. Furthermore, throughout the paper  we
 consider only \emph{reachable transducers}; that is, we assume 
that all the states of an initial transducer $\mathfrak A(s_0)$  are 
\emph{reachable}
from $s_0$: 
Given $s\in\Cal S$, there exists
input word $w$ over alphabet $\Cal I$ such that after the word $w$ has been
feeded to the automaton $\mathfrak A(s_0)$, the automaton reaches the state $s$. A reachable transducer is called \emph{finite}\ if its set $\Cal S$ of
states is finite, and is called \emph{infinite} if otherwise.

To  
 the initial automaton
$\mathfrak A(s_0)$ we put into a correspondence a family $\Cal F(\mathfrak
A)$ of all \emph{subautomata} $\mathfrak
A(s)=\langle\Cal I,\tilde{\Cal S},\Cal O,\tilde S,\tilde O, s\rangle$, $s\in\Cal S$, where
$\tilde{\Cal S}=\tilde{\Cal S}(s)\subset\Cal S$ is the set of all states that are reachable from the state
$s$ and $\tilde S,\tilde O$ are respective restrictions of the state transition and
output functions $S, O$ on $\Cal I\times \tilde{\Cal S}$.

Hereinafter in the paper the word `automaton' stands for an initial transducer whose input and output alphabet consists of $p$ symbols. We mostly assume that $p$ is
a prime although many of further results (for instance, the following Theorem
\ref{thm:auto-1L}) are true without this restriction
and 
so we identify input/output symbols with the $p$-element field $\F_p=\{0,1,\ldots,p-1\}$.
 Thus, for every $n=1,2,3,\ldots$ the automaton
$\mathfrak A(s_0)=\langle\F_p,\Cal S,\F_p,S,O,s_0\rangle$ maps $n$-letter words over $\F_p$ to $n$-letter words over
$\F_p$ according to the procedure described above,
cf. Figure \ref{fig:Transd-sc}.

We identify $n$-letter words over $\F_p$
with non-negative integers in a natural way: Given an $n$-letter word 
$w=\chi_{n-1}\chi_{n-2}\cdots\chi_0$ (i.e., $\chi_i\in\F_p$ for $i=0,1,2,\ldots,n-1$),
we consider $w$ as a base-$p$ expansion of the number $\chi_0+\chi_1\cdot
p+\cdots+\chi_{n-1}\cdot p^{n-1}$. In turn, the latter number can be considered
as an element of the residue ring $\Z/p^n\Z$ modulo $p^n$. Thus, to \emph{every
automaton $\mathfrak A$ there corresponds a map $f_{n,\mathfrak A}$ from $\Z/p^n\Z$ to $\Z/p^n\Z$}, for every $n=1,2,3,\ldots$.

Note that when necessary we may also identify $n$-letter words over $\F_p$
with elements of $\F_p^n$, the $n$-th Cartesian power of $\F_p$; so further
we use these one-to-one correspondences between $n$-letter words and residues
modulo $p^n$ (as well as between the words and elements from $\F_p^n$) without
extra comments.

In a similar manner, \emph{every automaton $\mathfrak A=\mathfrak A(s_0)$ defines  a map $f_{\mathfrak A}$ from  $\Z_p$ to $\Z_p$}:
Given an infinite word $w=\ldots\chi_{n-1}\chi_{n-2}\cdots\chi_0$ (that is,
an infinite  sequence) over $\F_p$ we consider
a $p$-adic integer whose $p$-adic canonical expansion is $z=z(w)=\chi_0+\chi_1\cdot
p+\cdots+\chi_{n-1}\cdot p^{n-1}+\cdots$; so, by the definition,  for every $z\in\Z_p$ we put
\begin{equation}
\label{eq:auto-f-def}
\delta_i(f_{\mathfrak A}(z))=O(\delta_i(z),s_i)\qquad  (i=0,1,2,\ldots),
\end{equation}
where $s_i=S(\delta_{i-1}(z),s_{i-1})$, $i=1,2,\ldots$, and $\delta_i(z)$
is the
$i$-th `$p$-adic digit' of $z$; that is,  the $i$-th term coefficient
in the $p$-adic canonical representation of $z$: $\delta_i(z)=\chi_i\in\F_p$,
$i=0,1,2,\ldots$.
The so defined map $f_{\mathfrak A}$ is called the \emph{automaton function}\index{automata function}
(or, the automaton map)
of the automaton $\mathfrak A$. 

The point is that the class of all automata functions that correspond to
automata with $p$-letter input/output alphabets coincides with the
class of all  maps from $\Z_p$ to $\Z_p$ that satisfy $p$-adic Lipschitz condition
with a constant 1 (1-Lipschitz maps, for brevity), see further Theorem \ref{thm:auto-1L}.
The claim of  the theorem is not a completely new result: Actually
the claim can be derived from a more
general result on asynchronous automata \cite[Proposition 3.7]{Grigorch_auto};
also, 
in a special case $p=2$ the claim was proved in  \cite{Vuillem_circ}. Nonetheless,
as we consider only synchronous
automata, however, for arbitrary $p$,
further in Section \ref{sec:auto-1L} we
 present a direct `$p$-adic' proof of Theorem \ref{thm:auto-1L}. Thus, to
 study automata functions one can apply various techniques of $p$-adic analysis.

We note that the $p$-adic approach (and wider the non-Archimedean one) has already been successfully applied to automata
theory. Seemingly the paper  \cite{Lunts} is the first one where  the $p$-adic techniques is applied to
study automata functions; the paper  deals with linearity
conditions of automata maps. For application of the non-Archimedean methods to automata
and formal
languages  see expository paper \cite{Pin_p-adic_auto} and references therein;
for applications to automata and group theory see \cite{Grigorch_auto,Grigorch_auto2_eng}.
In \cite{Vuillem_circ,Vuillem_DigNum,Vuillem_fin} the 2-adic methods are
used to study binary automata (the ones whose input/output alphabet is $\F_2$), in particular, to obtain the finiteness criterion for these automata. In monograph \cite{AnKhr} the $p$-adic ergodic theory
is studied (see numerous references therein) aiming at applications to computer
science  and cryptography (in particular, to automata theory, to pseudorandom
number generation and to stream cipher design) as well as to applications
in other areas like quantum theory, cognitive sciences and genetics.

The central result of the paper is the finiteness criterion for automata
in terms van der Put series of automata function, Theorem \ref{thm:fin-auto}. Compared to the finiteness
criterion from \cite{Vuillem_fin}, our's criterion 
\begin{itemize}
\item is more general ($p$ is arbitrary prime, and not only $p=2$ as in \cite{Vuillem_fin});
\item is more convenient, as the criterion from \cite{Vuillem_fin} is
in terms of values of functions $\delta_i(f_\mathfrak A)$, $i=0,1,2,\ldots$,
 which are difficult
(or even impossible)
to obtain from standard  $p$-adic representations of functions; whereas our's
directly uses one of these standard representations, the van der Put series.
\end{itemize}

The paper is organized as follows:
\begin{itemize}
\item In Section \ref{sec:auto-1L} we show that the class of all automata functions of automata with $p$-letter input/output alphabets 
coincides with the class of all 1-Lipschitz functions  defined on $\Z_p$
and valuated in $\Z_p$.
\item In Section \ref{sec:vdp} we remind basic facts about van der Put series;
for instance, the criterion if a function is 1-Lipschitz, in terms of its
van
der Put series.
\item In Section \ref{sec:main} we prove main result of the paper, the automaton
finiteness criterion in terms of van der Put series of the automaton function.
\item We conclude in Section \ref{sec:concl}.
\end{itemize}

\section{Automata functions = 1-Lipschitz functions}
\label{sec:auto-1L}
In the section, we prove that automata functions constitutes
the class of all 1-Lipschitz functions. We first remind a characterization
of 1-Lipschitz functions $f\:\Z_p\>\Z_p$ via properties of \emph{coordinate functions}; the latter are functions $\delta _{i}(f(x))$ defined on
$\Z_p$ and valuated in $\F_p=\{0,1,\ldots,p-1\}$. The $i$-th coordinate function
is merely a value of coefficient of the $i$-th term in a canonical $p$-adic expansion of $f(x)$.
\begin{prop}[cf. \protect{\cite[Proposition 3.35]{AnKhr}}]
\label{prop:compat_coord}
A function $f\colon{\mathbb Z}_{p}\rightarrow {\mathbb Z}_{p}$ is 1-Lipschitz if and only if for
every $i=1,2,\ldots  $ the $i$-th coordinate function $\delta _{i}(f(x))$ does not depend on
$\delta _{i+k}(x)$, for all $k=1,2,\ldots  $ .
\end{prop}

\begin{thm}
\label{thm:auto-1L}
The automaton function $f_{\mathfrak A(s_0)}\:\Z_p\>\Z_p$ of the automaton
$\mathfrak A(s_0)=\langle\F_p,\Cal S,\F_p,S,O,s_0\rangle$ is 1-Lipschitz.
Conversely, for every 1-Lipschitz function $f\:\Z_p\>\Z_p$ there exists an
automaton 
 $\mathfrak A(s_0)=\langle\F_p,\Cal S,\F_p,S,O,s_0\rangle$ such that  $f=f_{\mathfrak A(s_0)}$.
\end{thm}
\begin{proof}
As $s_i=S(\delta_{i-1}(z)),s_{i-1})$ for every $i=1,2\ldots$,  the $i$-th output symbol
$\xi_i=\delta_i(f_{\mathfrak A}(z))$ depends only on input symbols
$\chi_0,\chi_1,\ldots,\chi_i$; that is 
\[
\delta_i(f_{\mathfrak A}(z))=\psi_i(\delta_0(z),\delta_1(z),\ldots,\delta_i(z))
\]
for all $i=0,1,2,\ldots$ and for suitable maps $\psi_i\:\F_p^{i+1}\>\F_p$.
That is,
$f=f_{\mathfrak A(s_0)}\:\Z_p\>\Z_p$ is of the form
\begin{equation}
\label{eq:tri}
f\colon x=\sum_{i=0}^\infty\chi_ip^i\mapsto f(x)=\sum_{i=0}^\infty\psi_i(\chi_0,\ldots,\chi_i)p^i.
\end{equation}
By Proposition \ref{prop:compat_coord} this means that the function $f_{\mathfrak
A(s_0)}$ is 1-Lipschitz.

Conversely, let $f\:\Z_p\>\Z_p$ be a 1-Lipschitz map; the by Proposition \ref{prop:compat_coord} $f$ may be represented in the form \eqref{eq:tri}
for suitable maps $\psi_i\:\F_p^{i+1}\>\F_p$. We now construct an automaton  $\mathfrak A(s_0)=\langle\F_p,\Cal S,\F_p,S,O,s_0\rangle$ such that $f_{\mathfrak A(s_0)}=f$. 

Let  $\F_p^\star$ be a set of all non-empty words over the alphabet $\F_p$.
We consider these words  as base-$p$ expansions of numbers from $\N=\{1,2,3,\ldots\}$
and enumerate all these words  by integers $1,2,3,\ldots$ in lexicographical order in accordance with
the natural order on $\F_p$: $0<1<2<\cdots<p-1$. This way we establish a
one-to-one correspondence between the words $w\in\F_P^\star$ and integers $i\in\N$: $w \leftrightarrow \nu(w)$,  $i\leftrightarrow \omega(i)$ ($\nu(w)\in\N$, $\omega(i)\in\F_p^\star$).
Note that $\nu(\omega(i))=i$, $\omega(\nu(w))=w$ for all $i\in\N$ and all
non-empty words from
$w\in\F_p^\star$.
Assume that $\omega(0)$ is empty word.

Now put $\Cal S=\N_0=\{0,1,2,3,\ldots\}$, the set of all states of the automaton  $\mathfrak A$ under construction, and take  the initial state  $s_0=0$. The state transition function $S$ is defined as follows: 
\begin{equation}
\label{eq:auto-st}
S(r,i)=\nu(r\circ\omega(i)),
\end{equation}
where $i=0,1,2,\ldots$ and  $r\in\F_p$. That is,
$S(r,i)$ is the number of the word  $r\circ\omega(i)$ that is a concatenation
of the word   $\omega(i)$  (the word that has number $i$), the prefix, with the single-letter
word  $r$, the suffix. 

Now consider a one-to-one map 
$\theta_{n}(\chi_{n-1}\cdots\chi_1\chi_0)=(\chi_{0},\chi_1,\ldots,\chi_{n-1})$ from the $n$-letter words onto $\F_p^n$ and 
define the output function of the automaton $\mathfrak A$ as follows:
\begin{equation}
\label{eq:auto-out}
O(r,i)=\psi_{|\omega(i)|}(\theta_{|\omega(i)|+1}(r\circ\omega(i))),
\end{equation}
where $i=0,1,2,\ldots$ and  $r\in\F_p$.

As both $f$ and $f_{\mathfrak A(s_0)}$ are 1-Lipschitz, thus continuous with respect to $p$-adic metric, and as $\N_0$ is dense in $\Z_p$,
to prove that
 $f=f_{\mathfrak A(s_0)}$ is suffices to show that
 \begin{equation}
 \label{eq:auto-1}
 f_{\mathfrak
A(s_0)}(\tilde w)\equiv f(\tilde w)\pmod{p^{|w|)}}
\end{equation}
for all finite non-empty words $w\in\F_p^\star$, where $\tilde w\in\N_0$ stands for a integer whose base-$p$
expansion is $w$.  We prove that \eqref{eq:auto-1} holds for all $w\in\F_p^\star$
once  $|w|=n>0$ by induction on $n$.

If
$n=1$ then $\tilde w\in\F_p$; so once $w$ is feeded to
 $\mathfrak A$, the automaton reaches the state $S(w,0)=\nu(w)$ (cf. \eqref{eq:auto-st})
 and outputs
 $O(w,0)=\psi_0(\theta_1(w))=f(\tilde w)\md p$ (cf. \eqref{eq:auto-out}), see \eqref{eq:tri}.  Thus, \eqref{eq:auto-1}
 holds in this case.

Now assume that \eqref{eq:auto-1} holds for all $w\in\F_p^\star$ such that $|w|=n<k$ and prove that \eqref{eq:auto-1} holds also when $|w|=n=k$.
Represent $w=r\circ v$, where $r\in\F_p$ and $|v|=n-1$. By the induction
hypothesis, after the word   $v$ has been feeded to $\mathfrak A$, the automaton
reaches the state $\nu(v)$ and outputs the word $v_1$ of length $n-1$
such that
$\tilde v_1\equiv f(\tilde v)\md p^{n-1}$. Next, being feeded by the letter $r$, the
automaton (which is in the state $\nu(v)$ now) outputs the letter
$O(r,\nu(v))=\psi_{|\omega(\nu(v))|}(\theta_{|\omega(\nu(v))|+1}(r\circ\omega(\nu(v))))=
\psi_{|v|}(\theta_{|v|+1}(r\circ v))$.
This means that once feeded by $w$, the automaton $\mathfrak A(s_0)$ 
outputs the word
$v_2=(\psi_{|v|}(\theta_{|v|+1}(r\circ v)))\circ v_1$. However, $\tilde v_2\equiv
f(\tilde
w)\pmod{p^n}$, cf. \eqref{eq:tri}.
\end{proof}
\begin{note}
From the proof of Theorem \ref{thm:auto-1L} it is clear that the mapping
\[
f_{n,\mathfrak
A(s_0)}\:\Z/p^n\Z\>\Z/p^n\Z
\]
 is just a reduction  modulo $p^n$ of the automaton function $f_{\mathfrak
A(s_0)}$: $f_{n,\mathfrak
A(s_0)}=f_{\mathfrak
A(s_0)}\md p^n$ for all $n=1,2,3,\ldots$.
\end{note}
We remind that the \emph{reduction map} modulo $p^n$ is the map $\md p^n\:\Z_p\>\Z/p^n\Z$
such that $x\md p^n=\sum_{i=0}^{n-1}\delta_i(x)p^i$; that is, the map $\md
p^n$ just deletes all terms starting with  the $n$-th one in the canonical
$p$-adic expansion of $x$. Note that $\md p^n$ is a continuous ring epimorphism
of $\Z_p$ onto the residue ring $\Z/p^n\Z$. Thus, given a 1-Lipschitz map
$f\:\Z_p\>\Z_p$, the \emph{reduction modulo $p^n$}   
$f\md p^n:\Z/p^n\Z\>\Z/p^n\Z$ such that $(f\md p^n)(x)=f(x)\md p^n$ is well
defined by Proposition \ref{prop:compat_coord}; that is, the so defined map $f\md p^n$
does not depend on the choice of representatives in co-sets
with respect to the ideal $p^n\Z_p$. 

Further, given a 1-Lipschitz function $f\:\Z_p\>\Z_p$ via $\mathfrak A_f$
we denote an initial transducer $\langle\F_p,\Cal S,\F_p,S,O,s_o\rangle$ whose automaton function is $f$; that is, $f_{\mathfrak A_f}=f$. Note that the automaton $\mathfrak A_f$ is \emph{not} unique: There are many automata
that has the same automaton function. However, this non-uniqueness will not
cause misunderstanding since in the paper we are mostly interested with automata
functions rather than with `internal structure' (e.g., with state sets, state
transition and output functions, etc.) of automata themselves.

\section{Van der Put series of 1-Lipschitz functions}
\label{sec:vdp}
Now we remind  
definition and some properties of van der Put series, see e.g. \cite{Mah,Sch}
for  details.
Given a continuous 
function 
$f\: \Z_p\rightarrow \Z_p$, 
there exists 
a unique sequence 
$B_0,B_1,B_2, \ldots $ 
of $p$-adic integers such that 

\begin{equation}
\label{vdp}
f(x)=\sum_{m=0}^{\infty}
B_m \chi(m,x) 
\end{equation}
for all $x \in \Z_p$, where 
\begin{displaymath}
\chi(m,x)=\left\{ \begin{array}{cl}
1, &\text{if}
\ \left|x-m \right|_p \leq p^{-n} \\
0, & \text{otherwise}
\end{array} \right.
\end{displaymath}
and $n=1$ if $m=0$; $n$ is uniquely defined by the inequality 
$p^{n-1}\leq m \leq p^n-1$ otherwise. The right side series in \eqref{vdp} is called the \emph{van der Put series} of the function $f$. Note that
the sequence $B_0, B_1,\ldots,B_m,\ldots$ of \emph{van der Put coefficients} of
the function $f$ tends $p$-adically to $0$ as $m\to\infty$, and the series
converges uniformly on $\Z_p$. Vice versa, if a sequence $B_0, B_1,\ldots,B_m,\ldots$
of $p$-adic integers tends $p$-adically to $0$ as $m\to\infty$, then the the series in the right
part of \eqref{vdp} converges uniformly on $\Z_p$ and thus define a continuous
function $f\colon \Z_p\to\Z_p$.

The number $n$ in the definition of $\chi(m,x)$ has a very natural meaning;
it is just the number of digits in a base-$p$ expansion of $m\in\N_0$: 
\[
\left\lfloor  \log_p m \right\rfloor 
=
\left(\text{the number of digits in a base-} p \;\text{expansion for} \;m\right)-1;
\]
therefore $n=\left\lfloor  \log_p m \right\rfloor+1$ for all $m\in\N_0$ (that
is why we assume  $\left\lfloor  \log_p 0 \right\rfloor=0$). 

 Note that 
coefficients $B_m$ are
related to the values of the function $f$ in the following way:
Let 
$m=m_0+ \ldots +m_{n-2} p^{n-2}+m_{n-1} p^{n-1}$ be a base-$p$ expansion
for $m$, i.e., 
 $ m_j\in \left\{0,\ldots ,p-1\right\}$, $j=0,1,\ldots,n-1$ and $m_{n-1}\neq 0$, then
\begin{equation}
\label{eq:vdp-coef}
B_m=
\begin{cases}
f(m)-f(m-m_{n-1} p^{n-1}),\ &\text{if}\ m\geq p;\\
f(m),\ &\text{if otherwise}.
 
\end{cases}
\end{equation}
It worth noticing  also that $\chi (m,x)$ is merely  a characteristic function of the ball $\mathbf B_{p^{-\left\lfloor  \log_p m \right\rfloor-1}}(m)=m+p^{\left\lfloor  \log_p m \right\rfloor-1}\Z_p$
of radius $p^{-\left\lfloor  \log_p m \right\rfloor-1}$ centered at $m\in\N_0$:
\begin{equation}
\label{eq:chi}
\chi(m,x)=\begin{cases}
1,\ &\text{if}\ x\equiv m \pmod{p^{\left\lfloor  \log_p m \right\rfloor+1}};\\
0,\ &\text{if otherwise}
 
\end{cases}
 =
\begin{cases}
1,\ &\text{if}\ x\in\mathbf B_{p^{-\left\lfloor  \log_p m \right\rfloor-1}}(m);\\
0,\ &\text{if otherwise}
 
\end{cases}
\end{equation}

\begin{thm}[Anashin-Khrennikov-Yurova, \cite{AKY-DAN}]
\label{thm:vdp-comp}
The function $f\: \Z_p\rightarrow \Z_p$ is 1-Lipschitz if and only if
it can be represented as
\begin{equation}
\label{eq:vdp-comp}
f(x)=\sum_{m=0}^{\infty}b_m
p^{\left\lfloor \log_p m \right\rfloor} \chi(m,x),
\end{equation}
where $b_m\in \Z_p$ for $m=0,1,2,\ldots$
\end{thm}

\section{Main theorem}
\label{sec:main}
We first remind some notions and facts from the theory of automata sequences following
\cite{Allouche-Shall}.

An infinite sequence $\mathbf a=(a_i)_{i=0}^\infty$ over a finite alphabet $\Cal A$, $\#\Cal A=L<\infty$, is called \emph{$p$-automatic} if there
exists a finite transducer $\mathfrak T=\langle\F_p,\Cal S,\Cal A,S,O,s_0\rangle$ such
that for all $n=0,1,2,\ldots$, if $\mathfrak T$ is feeded by the word $\chi_k\chi_{k-1}\cdots\chi_0$
which is a base-$p$ expansion of $n=\chi_0+\chi_1p+\cdots\chi_kp^k$, $\chi_k\ne
0$ if $n\ne 0$, then the $k$-th output symbol of $\mathfrak T$ is $a_n$;
or, in other words, such that $\delta_{k}^{\Cal A}(f_\mathfrak T(n))=a_n$ for all $n\in\N_0$,
where $k=\lfloor\log_p n\rfloor$ 
and $\delta_k^{\Cal A}(r)$ stands for the $k$-th digit in the base-$L$ expansion of $r$.

The \emph{$p$-kernel} of the sequence $\mathbf a$ is the set $\ker_p(\mathbf
a)$ of all subsequences $(a_{jp^m+t})_{j=0}^\infty$, $m=0,1,2,\ldots$,
$0\le t<p^m$. 
\begin{thm}[Automaticity criterion, cf.
\protect{\cite[Theorem 6.6.2]{Allouche-Shall}}] 
\label{thm:p-ker}
Let $p\ge
2$; then the sequence $\mathbf a$ is $p$-automatic if and only if its $p$-kernel
is finite.
\end{thm}
Now we are able to state the main result of the paper:
\begin{thm}[Automata finiteness criterion]
\label{thm:fin-auto}
Given a 1-Lipschitz function $f\:\Z_p\>\Z_p$ represented by van der Put series
\eqref{eq:vdp-comp},
$$f(x)=\sum_{m=0}^{\infty}b_m
p^{\left\lfloor \log_p m \right\rfloor} \chi(m,x),$$ 
the function $f$ is the automaton function of a finite automaton
if and only if the following conditions hold simultaneously:
\begin{enumerate}
\item all coefficients $b_m$, $m=0,1,2,\ldots$, constitute a finite subset $B_f\subset\Q\cap\Z_p$,
and
\item the $p$-kernel of the sequence $(b_m)_{m=0}^\infty$ is finite.
\end{enumerate}
\end{thm}

\begin{note} 
Condition (ii) of the theorem is equivalent to the condition that
the  sequence $(b_m)_{m=0}^\infty$ is $p$-automatic, cf. Theorem \ref{thm:p-ker}.
\end{note}
Now we are going to present equivalent statement of Theorem \ref{thm:fin-auto},
in terms of formal power series.

Given a $q$-element field $\F_q$, denote via $\F_q[[X]]$ the ring of formal
power series in variable $X$ over $\F_q$: 
\[
\F_q[[X]]=\left\{\sum_{i=0}^\infty a_iX^i\: a_i\in\F_q\right\};
\]
denote via $\F_q((X))$ the ring of formal Laurent series over $\F_q$:
\[
\F_q((X))=\left\{\sum_{i=-n_0}^\infty a_iX^i\: n_0\in\N_0, a_i\in\F_q\right\}.
\]
Denote via $F_q(X)$ the field of (univariate) rational functions over $\F_q$:
\[
\F_q(X)=\left\{\frac{u(X)}{v(X)}\: u(X),v(X)\in\F_p[X], u(X)\ne 0\right\},
\]
where $\F_p[X]$ is the ring of polynomials in variable $X$ over $\F_q$. As the field $\F_q((X))$ contains a subfield $\F_q(X)$, it is possible
to define algebraicity over $\F_q(X)$:
A formal Laurent series $F(X)=\sum_{i=-n_0}^\infty a_iX^i$ is  \emph{algebraic}
over $\F_p(X)$ if and only if there exist $d\in\N$ and polynomials 
$u_0(X),\ldots, u_{d}(X)\in\F_p[X]$, not all zero, such that in the field
$\F_q((X))$ the following identity holds:
\[
u_0(X)+u_1(X)\cdot F(X)+\cdots+u_{d}(X)\cdot(F(X))^d=0.
\]

Now we remind Christol's theorem, \cite[Theorem 12.2.5]{Allouche-Shall}:
\begin{thm}[Christol] Let $p$ be a prime, and let  $\mathbf a=(a_i)_{i=0}^\infty$ be infinite sequence over a finite non-empty alphabet $\Cal A$. The sequence
$\mathbf a$ is $p$-automatic if and only if there exists an integer $\ell\in\N$
and an injection $\tau:\Cal
A\>\F_{p^\ell}$ such that the formal power series $\sum_{i=0}^\infty\tau(a_i)X^i$
is algebraic over $\F_{p^\ell}(X)$.
\end{thm}
By Christol's theorem, we now may replace condition (ii) from the statement
of Theorem \ref{thm:fin-auto} by equivalent one, thus getting an equivalent
finiteness criterion:
\begin{thm}[Automata finiteness criterion, equivalent]
\label{thm:fin-auto-e}
Given a 1-Lipschitz function $f\:\Z_p\>\Z_p$ represented by van der Put series
\eqref{eq:vdp-comp}, 
the function $f$ is the automaton function of a finite automaton
if and only if the following conditions hold simultaneously:
\begin{enumerate}
\item all coefficients $b_m$, $m=0,1,2,\ldots$, constitute a finite subset $B_f\subset\Q\cap\Z_p$,
and
\item under a suitable injection $\tau\:B_f\to\F_{p^\ell}$, the formal power series
$$
\sum_{m=0}^\infty\tau(b_m)X^m
$$
over $\F_{p^\ell}$ is algebraic over $\F_{p^\ell}(X)$.
\end{enumerate}
\end{thm}


\begin{proof}[Proof of Theorem \ref{thm:fin-auto}]
Given a 1-Lipschitz function $f$, for $n\in\N_0$, $k\ge\lfloor\log_pn\rfloor+1$ consider 
functions $f_{n,k}\:\Z_p\>\Z_p$ defined as follows:
$$
f_{n,k}(z)=\frac{1}{p^{k}}\left(f(n+p^{k}z)-(f(n)\bmod p^{k})\right);\ z\in\Z_p.
$$
The function $f$ is an automaton function of a finite automaton if and only if in
the collection $\Cal F$ of functions $f_{n,k}$ ($n\in\N_0$, $k\in\N$, $k\ge\lfloor\log_pn\rfloor+1$) contains only finitely many pairwise
distinct functions: Note that  $f_{n,k}$ is the automaton function that corresponds
to the automaton $\mathfrak A(s(n_k))=\langle\F_p,\Cal S,\F_p,S,O, s(n_k)\rangle$,
where $s(n_k)\in\Cal S$ is the state the automaton $\mathfrak A=\mathfrak A_f=\langle\F_p,\Cal S,\F_p,S,O, s_0\rangle$ reaches after it has been feeded with the input word $n_k$ (of length $p^{k}$) that corresponds to a base-$p$
expansion of $n$ (so the word $n_k$ may contain some leading zeros that correspond
to higher order digits of the expansion). 

Note that by Theorem \ref{thm:vdp-comp}, $b_{n+p^ks}=\frac{1}{p^k}(f(n+p^ks)-f(n))=
\frac{1}{p^k}(f(n+p^ks)-(f(n)\bmod p^k))-\frac{1}{p^k}(f(n)-(f(n)\bmod
p^k))=f_{n,k}(s)-f_{n,k}(0)$
if $n\le p^k-1$ and $s\in\{1,2,\ldots,p-1\}$, cf. \eqref{eq:vdp-coef}; so
the finiteness of $\Cal F$ implies that in the sequence $(b_m)_{m=0}^\infty$
there are only finitely many pairwise distinct terms. We proceed with this in mind.

 Take $n\in\N_0$ and $k\ge\lfloor\log_pn\rfloor+1$. By \eqref{eq:vdp-comp}, the value $f(n+p^kz)$ can be represented as $f(n+p^kz)=A_{n,k}(z)+B_{n,k}(z)$, where
\begin{align}
A_{n,k}(z)=&\sum_{m=0}^{p^k-1}b_mp^{\lfloor\log_pm\rfloor}\chi(m,n+p^kz);
\label{eq:A}\\
B_{n,k}(z)=&\sum_{m=p^k}^\infty b_mp^{\lfloor\log_pm\rfloor}\chi(m,n+p^kz).
\label{eq:B}
\end{align}
By  \eqref{eq:chi}, once $m\le p^k-1$, the equality $\chi(m,n+p^kz)=0$
holds if
and only if $m\not\equiv n\pmod{p^{\lfloor\log_pm\rfloor+1}}$; and once $m\ge p^k$, the equality $\chi(m,n+p^kz)=0$ holds if and only if $m\not\equiv n\pmod
{p^{\lfloor\log_pm\rfloor+1}}$ (note that $\lfloor\log_pm\rfloor+1\ge k+1$
under conditions of  the latter case);
thus
\begin{align}
A_{n,k}(z)=&\sum_{m=0}^{p^k-1}b_mp^{\lfloor\log_pm\rfloor}\chi(m,n);
\label{eq:A1}\\
B_{n,k}(z)=&\sum_{t=1}^\infty b_{n+p^kt}p^{k+\lfloor\log_pt\rfloor}\chi(n+p^kt,n+p^kz)=
\sum_{t=1}^\infty b_{n+p^kt}p^{k+\lfloor\log_pt\rfloor}\chi(t,z).\label{eq:B1}
\end{align} 
From here in particular it follows that $A_{n,k}(z)$ does not depend on $z$ and
that $B_{n,k}(z)\equiv 0\pmod{p^k}$ for all $z\in\Z_p$; consequently,
\begin{equation}
\label{eq:f}
f_{n,k}(z)=\frac{1}{p^{k}}
\left(A_{n,k}(0)+B_{n,k}(z)-(A_{n,k}(0)\bmod p^{k})\right)=
C_{n,k}+D_{n,k}(z),
\end{equation}
where
\begin{align}
\label{eq:C}
C_{n,k}=&\frac{1}{p^{k}}
\left(A_{n,k}(0)-(A_{n,k}(0)\bmod p^{k})\right);\\
D_{n,k}(z)=&\sum_{t=1}^\infty b_{n+p^kt}p^{\lfloor\log_pt\rfloor}\chi(t,z).
\label{eq:D}
\end{align}
Note that from \eqref{eq:D} it follows that $D_{n,k}(0)=0$ (cf. \eqref{eq:chi}), so from \eqref{eq:f}
we deduce  that $f_{n,k}(0)=C_{n,k}$.

Thus, we have obtained the following criterion of the finiteness of the number
of distinct functions in the collection $\Cal F$:
There are only finite number of
pairwise distinct functions $f_{n,k}\in\Cal F$, $n\in\N_0$, $k\ge\lfloor\log_pn\rfloor+1$, if and only if the following two conditions hold simultaneously:
\renewcommand{\theenumi}{\arabic{enumi}}
\renewcommand{\labelenumi}{\theenumi.} 
\begin{enumerate}
\item There are only
finitely many pairwise distinct constants $C_{n,k}\in\Z_p$, $n\in\N_0$, $k\ge\lfloor\log_pn\rfloor+1$.
\item There are only finitely many pairwise distinct
functions $D_{n,k}\:\Z_p\>\Z_p$, $n\in\N_0$, $k\ge\lfloor\log_pn\rfloor+1$.
\end{enumerate}
 However, since representation \eqref{eq:D}
is a (unique) van der Put expansion of the function $D_{n,k}$,  condition
2 
is equivalent to the condition that in the sequence
$(b_n)_{n=0}^\infty$ there are only finitely many pairwise distinct subsequences
$(b_{n+p^kt})_{t=1}^\infty$, where $k\ge\lfloor\log_pn\rfloor+1$, $n\in\N_0$.
In turn, the latter condition is equivalent to the condition  that there are only finitely many pairwise distinct subsequences
$(b_{n+p^kt})_{t=0}^\infty$, $n\in\N_0$, $k\in\N$. Note that if the condition
holds, there are only finitely many pairwise distinct terms $b_m$ in the
sequence $(b_m)_{m=0}^\infty$.

Consider condition 1. Note that  for $k\ge\lfloor\log_pn\rfloor+1$ we
have that
\begin{equation}
\label{eq:short-sum}
\sum_{m=0}^{p^k-1}b_mp^{\lfloor\log_pm\rfloor}\chi(m,n)=
\sum_{m=0}^{p^{\lfloor\log_pn\rfloor+1}-1}b_mp^{\lfloor\log_pm\rfloor}\chi(m,n)
\end{equation}
since $\chi(m,n)=0$ once $\lfloor\log_pm\rfloor>\lfloor\log_pn\rfloor$, cf.
\eqref{eq:chi}. Thus, $A_{n,k}(z)$ does not depend on $k$ (and on $z$ as
we have already shown); so denoting the right hand side in \eqref{eq:short-sum}
via $A(n)$, we have that $C_{n,k}=p^{-k}(A(n)-((A(n))\bmod p^k))$, cf. \eqref{eq:C}
and \eqref{eq:A1}. 

From here by \eqref{eq:chi} we get that
$C_{n,k}=p^{-k}(b_n-(b_n\bmod p^k))$
once $n$ is such that $\lfloor\log_pn\rfloor=0$. Consequently, 
given $n\in\{0,1,\ldots,p-1\}$,
the finiteness of the number of pairwise distinct $C_{n,k}$, where
$k\ge\lfloor\log_p0\rfloor+1=1$,
is equivalent to the condition that the sequence $(\delta_i(b_n))_{i=0}^\infty$
is eventually periodic, that is, to the condition that $b_n\in\Q\cap\Z_p$. Using
this as a base for induction on $\lfloor\log_pn\rfloor$, assuming that
all $b_n\in\Q\cap\Z_p$ for $n$ such that $\lfloor\log_pn\rfloor< N$, 
we
see that, given $n$ such that $\lfloor\log_pn\rfloor=N$,
the finiteness of the number
of pairwise distinct $C_{n,k}=p^{-k}(A(n)-((A(n))\bmod p^k))$ for $k\ge\lfloor\log_pn\rfloor+1=N+1$ is equivalent
to the condition that $A(n)\in\Q\cap\Z_p$. However, in view of  \eqref{eq:chi}
from the definition of
$A(n)$ it follows that 
$$
A(n)=b_np^{\lfloor\log_pn\rfloor}+\sum_{m=0}^{p^{\lfloor\log_pn\rfloor}-1}b_mp^{\lfloor\log_pm\rfloor}\chi(m,n).
$$
The left side sum is in $\Q\cap\Z_p$ by induction hypothesis; so $b_np^{\lfloor\log_pn\rfloor}\in\Q\cap\Z_p$
and whence $b_n\in\Q\cap\Z_p$ since $b_n\in\Z_p$. 

We finally have prowed that  conditions 1--2 hold simultaneously if and only
if the following conditions hold simultaneously:
\renewcommand{\labelenumi}{\theenumi$^\prime$.}
\begin{enumerate}
\item All coefficients $b_m$, $m=0,1,2,\ldots$, constitute a non-empty finite subset
(denoted as $B_f$)
in $\Q\cap\Z_p$.
\item There are only finitely many pairwise distinct subsequences
$(b_{n+p^kt})_{t=0}^\infty$, $n\in\N_0$, $k\in\N$, $k\ge\lfloor\log_pn\rfloor+1$;
that is, the $p$-kernel of the sequence $(b_m)_{m=0}^\infty$ is finite.
\end{enumerate}
This ends the proof.
\end{proof}

\section{Discussion}
\label{sec:concl}
In the paper, we found finiteness conditions for an automaton in terms of van
der Put series of the automaton function. Any automaton function of a transducer
with $p$-letter input/output alphabets can be considered as a continuous
(with respect to the $p$-adic distance)
map from $p$-adic integers to $p$-adic integers. The van der Put series is
a standard representation for continuous $p$-adic maps. The paper discloses
relations between the $p$-adic theory and the theory of automata sequences.



\begin{thebibliography}{99}

\bibitem{Allouche-Shall}
J.-P. Allouche and J.~Shallit.
\newblock {\em Automatic Sequences. Theory, Applications, Generalizations}.
\newblock Cambridge Univ. Press, 2003.

\bibitem{AnKhr}
V.~Anashin and A.~Khrennikov.
\newblock {\em Applied Algebraic Dynamics}, volume~49 of {\em de Gruyter
  Expositions in Mathematics}.
\newblock Walter~de~Gruyter GmbH \& Co., Berlin---N.Y., 2009.

\bibitem{AKY-DAN}
V.~S. Anashin, A.~Yu. Khrennikov, and E.~I. Yurova.
\newblock Characterization of ergodicity of $p$-adic dynamical systems by using
  van der {P}ut basis.
\newblock {\em Doklady Mathematics}, 83(3):306--308, 2011.

\bibitem{Bra}
W.~Brauer.
\newblock {\em Automatentheorie}.
\newblock B.~G.~Teubner, Stuttgart, 1984.

\bibitem{Grigorch_auto2_eng}
R.~I. Grigorchuk.
\newblock Some topics in the dynamics of group actions on rooted trees.
\newblock {\em Proc. Steklov Institute of Mathematics}, 273:64--175, 2011.

\bibitem{Grigorch_auto}
R.~I. Grigorchuk, V.~V. Nekrashevich, and V.~I. Sushchanskii.
\newblock Automata, dynamical systems, and groups.
\newblock {\em Proc. Steklov Institute Math.}, 231:128--203, 2000.

\bibitem{Hensel}
K.~Hensel.
\newblock \"{U}ber eine neue {B}egr\"undung der {T}heorie der algebraischen
  {Z}ahlen.
\newblock {\em Jahresbericht der Deutschen Mathematiker-Vereinigung},
  6(3):83--88, 1897.

\bibitem{Lunts}
A.~G. Lunts.
\newblock The $p$-adic apparatus in the theory of finite automata.
\newblock {\em Problemy Kibernetiki}, 14:17--30, 1965.
\newblock In Russian.

\bibitem{Mah}
K.~Mahler.
\newblock {\em $p$-adic numbers and their functions}.
\newblock Cambridge Univ. Press, 1981.
\newblock (2nd edition).

\bibitem{Pin_p-adic_auto}
J.-E. Pin.
\newblock Profinite methods in automata theory.
\newblock In {\em Symposium on Theoretical Aspects of Computer Science ---
  STACS 2009}, pages 31--50, Freiburg, 2009.

\bibitem{Sch}
W.~H. Schikhof.
\newblock {\em Ultrametric calculus}.
\newblock Cambridge University Press, 1984.

\bibitem{VL}
V.~S. Vladimirov and I.~V. Volovich.
\newblock Superanalysis 1. {D}ifferential calculus.
\newblock {\em Teoret. Mat. Fiz.}, 59:3--27, 1984.

\bibitem{VL1}
V.~S. Vladimirov and I.~V. Volovich.
\newblock Superanalysis 2. {I}ntegral calculus.
\newblock {\em Teoret. Mat. Fiz.}, 60:169--198, 1984.

\bibitem{Vladimirov/Volovich/Zelenov:1994}
V.~S. Vladimirov, I.~V. Volovich, and E.~I. Zelenov.
\newblock {\em $p$-adic Analysis and Mathematical Physics}.
\newblock Scientific, Singapore, 1994.

\bibitem{Volovich:1987}
I.~V. Volovich.
\newblock $p$-adic string.
\newblock {\em Class. Quant. Grav.}, 4:83--87, 1987.

\bibitem{Vuillem_circ}
J.~Vuillemin.
\newblock On circuits and numbers.
\newblock {\em IEEE Trans. on Computers}, 43(8):868--879, 1994.

\bibitem{Vuillem_fin}
J.~Vuillemin.
\newblock Finite digital synchronous circuits are characterized by 2-algebraic
  truth tables.
\newblock In {\em Advances in computing science - ASIAN 2000}, volume 1961 of
  {\em Lecture Notes in Computer Science}, pages 1--7, 2000.

\bibitem{Vuillem_DigNum}
J.~Vuillemin.
\newblock Digital algebra and circuits.
\newblock In {\em Verification:Theory and Practice}, volume 2772 of {\em
  Lecture Notes in Computer Science}, pages 733--746, 2003.

\end{thebibliography}

\end{document}